\documentclass[11pt,a4paper]{article}

\usepackage[T1]{fontenc}
\usepackage[utf8]{inputenc}
\usepackage{lmodern}
\usepackage[margin=26mm]{geometry}
\usepackage{amsmath,amssymb,amsthm,mathtools}
\usepackage{booktabs,array}
\usepackage{tabularx}
\usepackage{microtype}
\usepackage[table]{xcolor}
\usepackage{arydshln}
\usepackage[hidelinks]{hyperref}
\usepackage{enumitem}
\usepackage{placeins}
\usepackage{needspace}
\usepackage{tikz}
\usetikzlibrary{arrows.meta}

\setlist{nosep,leftmargin=1.5em}
\allowdisplaybreaks

\newcommand{\E}{\mathbb{E}}
\newcommand{\N}{\mathcal N}
\newcommand{\1}{\mathbf 1}
\newcommand{\VCG}{\mathrm{VCG}}
\newcommand{\AGV}{\mathrm{AGV}}
\newcommand{\GUVCG}{\mathrm{GU\text{-}VCG}}
\newcommand{\GUM}{\mathrm{GUM}}
\DeclareMathOperator*{\argmax}{arg\,max}
\newtheorem{proposition}{Proposition}
\newtheorem{conjecture}{Conjecture}

\definecolor{cblue}{RGB}{0,0,210}
\definecolor{dpurple}{RGB}{180,0,180}
\definecolor{llgreen}{RGB}{222,244,222}
\definecolor{llred}{RGB}{250,224,224}
\definecolor{cgreen}{RGB}{0,125,0}
\definecolor{cred}{RGB}{190,0,0}
\definecolor{guegreen}{RGB}{32,120,70}
\definecolor{guered}{RGB}{170,40,40}
\definecolor{dsebg}{RGB}{232,248,232}
\definecolor{bnebg}{RGB}{255,244,244}
\definecolor{guebg}{RGB}{184,242,184}

\newcommand{\featureon}[2]{\textcolor{cgreen}{\(\mathsf{#1}^{+}\)~\textbf{#2}}}
\newcommand{\featureoff}[2]{\textcolor{cred}{\(\mathsf{#1}^{-}\)~\textbf{#2}}}
\newcommand{\featureonmark}[1]{\textcolor{cgreen}{\(\mathsf{#1}^{+}\)}}
\newcommand{\featureoffmark}[1]{\textcolor{cred}{\(\mathsf{#1}^{-}\)}}
\newcommand{\yesprop}{\cellcolor{llgreen}\textcolor{cgreen}{\textbf{Yes}}}
\newcommand{\noprop}{\cellcolor{llred}\textcolor{cred}{\textbf{No}}}
\newcommand{\DSEmark}{\cellcolor{dsebg}\textcolor{guegreen}{\textbf{DSE}}}
\newcommand{\BNEmark}{\cellcolor{bnebg}\textcolor{black}{\textbf{BNE}}}
\newcommand{\PBEmark}{\cellcolor{bnebg}\textcolor{black}{\textbf{PBE}}}
\newcommand{\PEEmark}{\cellcolor{bnebg}\textcolor{black}{\scriptsize\textbf{periodic ex post}}}
\newcommand{\GUEmark}{\cellcolor{guebg}\textcolor{guegreen!60!black}{\textbf{GUE}}}
\newcommand{\rfit}[2]{\phantom{#1}\mathllap{#2}}
\newcommand{\tuplefit}[6]{(\rfit{#1}{#4},\,\rfit{#2}{#5},\,\rfit{#3}{#6})}
\newcommand{\tuplefitcompact}[6]{(\rfit{#1}{#4},\rfit{#2}{#5},\rfit{#3}{#6})}
\newcommand{\reporttriplet}[3]{\tuplefit{-6}{-6}{-6}{#1}{#2}{#3}}
\newcommand{\reporttripletpos}[3]{\tuplefit{10}{-6}{-6}{#1}{#2}{#3}}
\newcommand{\lefttransfertriplet}[3]{\tuplefit{-2.5}{-3.5}{-4}{#1}{#2}{#3}}
\newcommand{\righttransfertriplet}[3]{\tuplefit{-1.5}{-4.5}{-4}{#1}{#2}{#3}}
\newcommand{\pairfit}[4]{(\rfit{#1}{#3},\,\rfit{#2}{#4})}
\newcommand{\historyleft}[1]{\makebox[10.5em][l]{#1}}
\newcommand{\historycenter}[1]{\makebox[10.5em][c]{#1}}
\newcommand{\alignednum}[2]{\mathmakebox[1.25em][r]{#1}\mathmakebox[0.80em][l]{#2}}
\newcommand{\faintdottedrule}{\arrayrulecolor{black}\hdashline[0.8pt/1.8pt]\arrayrulecolor{black}}

\title{From the Social Choice Problem to a Collusion-Proof\\
Tendering Mechanism for Dynamic Stochastic Projects}
\author{Endre Cs\'oka\thanks{Alfréd Rényi Institute of Mathematics, Budapest, Hungary. This note draws on joint work with Heng Liu, Andr\'as Pongr\'acz, Alexander Rodivilov, Alexander Teytelboym, and M\'at\'e Viczi\'an, and has benefited from their contributions and assistance.
This work was supported by the NRDI grant KKP 138270.
}}
\date{}

\hypersetup{
  pdftitle={From the Social Choice Problem to a Collusion-Proof Tendering Mechanism for Dynamic Stochastic Projects},
  pdfauthor={Endre Csoka}
}

\begin{document}
\maketitle

\begin{abstract}
The VCG family \cite{Vickrey1961,Clarke1971,Groves1973} and the AGV mechanism \cite{Arrow1979,DAspremontGerardVaret1979} are two classical approaches to efficient implementation in the static social choice problem.
In 2024, Cs\'oka et al. \cite{CsokaEtAl2026} showed that AGV has critical weaknesses.
In contrast, the transferable-utility Guaranteed Utility Mechanism (TU-GUM) \cite{CsokaEtAl2026} retains all the standard desirable properties of AGV while adding further ones, including collusion-proofness, because it implements efficiency in Guaranteed Utility Equilibrium \cite{CsokaPongraczRodivilov2025}.
TU-GUM also applies to a more general dynamic setting with multiple extensions.
Moreover, TU-GUM is a special case of an even more general and robust mechanism \cite{Csoka2015} that combines contingent first-price tendering with the coordinated execution of dynamic stochastic multi-agent projects through a surprisingly simple rule.
This paper summarizes and connects existing results from a different perspective, with some minor new observations.
\end{abstract}

\section{The static social choice problem: VCG, AGV, and TU-GUM}

\paragraph{The static social choice problem.}
Let \(\N=\{1,\ldots,n\}\) be the set of agents.
The type profile \(\theta=(\theta_i)_{i\in\N}\) is drawn from a prior \(\mu\), and each realized type \(\theta_i\) is privately observed by agent \(i\).
The agents simultaneously report \(\widehat\theta_i\).
For a decision \(x\in X\), agent \(i\)'s payoff is \(u_i(\theta_i,x)\), and the transfer to agent \(i\) is denoted by \(y_i\), so his utility is \(u_i(\theta_i,x)+y_i\).
Throughout the paper, payoff excludes transfers, whereas utility includes them.
The objective is to implement an efficient social choice while providing incentives for truthful reporting.
A mechanism is \emph{budget-balanced} if \(\sum\limits_{i\in\N}y_i(\widehat\theta)=0\) for every report profile.

A mechanism is \emph{prior-free} if its rules do not depend on the prior \(\mu\).
Unless stated otherwise, for prior-dependent mechanisms we assume independent types, namely, \(\mu=\bigotimes_{i\in\N}\mu_i\).

\paragraph{Efficient decision policy and externality notation.}
Fix an efficient decision policy \(\chi\) satisfying
\[
  \chi(\widehat\theta)
  \in \argmax_{x\in X}\sum_{k\in\N}u_k(\widehat\theta_k,x)
\]
We study the implementation of this fixed policy.
To describe AGV and TU-GUM, we use the following notation.
Let \(\widetilde\theta\) be drawn from \(\mu\), independently of the realized type profile \(\theta\) and of the reports.
For \(B\subseteq\N\), write \(\bar B:=\N\setminus B\), and let \((\widehat\theta_B,\widetilde\theta_{\bar B})\) denote the combined profile. Define
\[
 V_j(\widehat\theta_B)
 :=\E_{\widetilde\theta_{\bar B}}
 \Bigl[u_j\bigl((\widehat\theta_B,\widetilde\theta_{\bar B})_j,
 \chi(\widehat\theta_B,\widetilde\theta_{\bar B})\bigr)\Bigr]
\]
If \(i\notin B\), the bilateral marginal externality of fixing agent \(i\)'s report, given that the reports in \(B\) have already been fixed, is
\begin{equation}
 \gamma^{i\to j}_B
 :=V_j\bigl(\widehat\theta_{B\cup\{i\}}\bigr)
   -V_j(\widehat\theta_B)
 \label{eq:gammaB}
\end{equation}
This is the change in agent \(j\)'s anticipated payoff caused by processing agent \(i\)'s report at that stage.

\paragraph{GUE implementation.}
Here and throughout this static section, expected-utility guarantees are \emph{ex ante}, before the type profile is realized.
A \emph{GUE implementation} of a socially efficient outcome means that, regardless of the other agents' strategies, truthful reporting gives each agent ex-ante expected utility at least as high as under the truthful (and hence efficient) strategy profile.

\subsection{The VCG family}

The VCG family associated with the efficient decision policy \(\chi\) consists of the Groves transfer rules \cite{Groves1973}:
\begin{equation}
 y_i^{\VCG,h}(\widehat\theta)
 :=\sum_{j\ne i}u_j\bigl(\widehat\theta_j,\chi(\widehat\theta)\bigr)
   +h_i(\widehat\theta_{-i})
 \label{eq:VCGgeneral}
\end{equation}
where \(h_i\) is an arbitrary function of the other agents' reports.
VCG is a prior-free family of dominant-strategy implementations of \(\chi\), none of which is budget-balanced.
Different choices of the functions \(h_i\) give different payment rules within the VCG family.
Appendix~\ref{app:guvcg} illustrates several natural choices, including a new variant that is also a GUE implementation under additional conditions.

\subsection{AGV}

AGV uses the expected externality generated by an agent's report \cite{Arrow1979,DAspremontGerardVaret1979}:
\begin{equation}
 G_i(\widehat\theta_i)
 :=\sum_{j\ne i}\gamma^{i\to j}_{\varnothing}
 =\sum_{j\ne i}\Bigl(V_j(\widehat\theta_i)-V_j(\varnothing)\Bigr)
 \label{eq:Gi}
\end{equation}
The AGV transfer is
\begin{equation}
 y_i^{\AGV}(\widehat\theta)
 :=G_i(\widehat\theta_i)
 -\frac{1}{n-1}\sum_{k\ne i}G_k(\widehat\theta_k)
 \label{eq:AGV}
\end{equation}
Thus agent \(i\) receives the expected externality caused by his report, while the corresponding amount is paid by the other agents in equal shares.
The transfers balance exactly: \(\sum_i y_i^{\AGV}=0\).
Truthful reporting is a Bayesian Nash equilibrium.

\subsection{TU-GUM}

We introduce \emph{sequential bilateral externalities} in the main paper \cite{CsokaEtAl2026} and construct TU-GUM from them.
In the present static setting, TU-GUM processes the agents in the fixed order \(1,2,\ldots,n\).
Let
\[
 B_i:=\{1,2,\ldots,i-1\}
\]
When \(i\)'s report is processed, its sequential bilateral externality on \(j\) is
\begin{equation}
 \gamma^{i\to j}(\widehat\theta)
 :=\gamma_{B_i}^{i\to j}
 \label{eq:seqgamma}
\end{equation}
Agent \(j\) pays this externality to agent \(i\).
Consequently, TU-GUM is budget-balanced, with net transfer to agent \(i\)
\begin{equation}
 y_i^{\GUM}(\widehat\theta)
 :=\sum_{j\ne i}
 (\gamma^{i\to j}-\gamma^{j\to i})
 \label{eq:seqGUM}
\end{equation}
Truthful reporting provides the guarantee in the definition above, so TU-GUM implements \(\chi\) in GUE.
Averaging its payment rule over all possible orders gives symmetrized TU-GUM, preserving budget balance and the same utility guarantees, and hence the GUE implementation property.

\subsection{Comparison of VCG, AGV, and TU-GUM}

AGV and TU-GUM differ in two respects.
First, AGV evaluates each externality before any other report has been fixed, with \(B=\varnothing\), whereas TU-GUM processes reports sequentially and evaluates each bilateral externality conditional on the reports already processed.
This difference is already present with two agents.
Second, AGV aggregates the externalities caused by an agent's report and charges the aggregate to the other agents in equal shares, whereas TU-GUM settles each bilateral externality directly between the two agents concerned.
This second difference disappears with two agents.

The central static comparison is summarized in Table~\ref{tab:mechanism-properties}.
DSE, BNE, and PBE denote dominant-strategy equilibrium, Bayesian Nash equilibrium, and perfect Bayesian equilibrium, respectively; IEWDS denotes iterated elimination of weakly dominated strategies.

\begin{table}[!htbp]
\centering
\caption{Main properties of the VCG family, AGV, and TU-GUM.}
\label{tab:mechanism-properties}
\small
\renewcommand{\arraystretch}{1.16}
\setlength{\tabcolsep}{3pt}
\begin{tabularx}{\textwidth}{@{}>{\raggedright\arraybackslash}X *{3}{>{\centering\arraybackslash}p{0.15\textwidth}}@{}}
\toprule
Property & VCG & AGV & TU-GUM\\
\midrule
\multicolumn{4}{c}{\cellcolor{gray!15}\textit{Efficiency and strategic robustness}}\\
Efficiency is implemented in & \DSEmark & \BNEmark & \GUEmark\\
Efficiency survives IEWDS & \yesprop & \noprop & \yesprop\\
All BNEs are efficient & \noprop & \noprop & \yesprop\\
Collusion-proof & \noprop & \noprop & \yesprop\\
\midrule
\multicolumn{4}{c}{\cellcolor{gray!15}\textit{Other properties}}\\
Balanced budget & \noprop & \yesprop & \yesprop\\
Prior-free & \yesprop & \noprop & \noprop\\
\bottomrule
\end{tabularx}
\end{table}
\FloatBarrier

On unrestricted quasilinear domains, the Groves rules in \eqref{eq:VCGgeneral} are exactly the transfer rules that implement the efficient decision policy \(\chi\) with truthful reporting as a dominant strategy \cite{GreenLaffont1977}.
The Green--Laffont impossibility theorem further shows that such implementation cannot in general be combined with exact budget balance \cite{GreenLaffont1979}.

Appendix~\ref{app:agvfragility} presents a social-choice setup in which exhaustive IEWDS under AGV eliminates every efficient strategy profile, independently of the elimination order.

For an inefficient BNE common to the VCG family, consider a Vickrey auction with \(\Theta_1=\{0,2\}\) and \(\Theta_2=\{1,3\}\), where agent 1 always reports \(0\) and agent 2 always reports \(3\).

The contrast becomes sharper in dynamic settings.
Static VCG implements efficiency in dominant strategies.
Its dynamic analogue, the Dynamic Pivot Mechanism (DPM) of Bergemann and V\"alim\"aki \cite{BergemannValimaki2010}, instead implements the efficient policy in what they call periodic ex post equilibrium.
This periodic ex post implementation is not robust to weak dominance: Proposition~1 of \cite{CsokaEtAl2026} gives a setting in which IEWDS eliminates every efficient strategy profile, including the truthful periodic ex post equilibrium, regardless of the elimination order.
In contrast, GUE implementation extends to the dynamic setting, and therefore so do all the important features of TU-GUM.

\begin{table}[!htbp]
\centering
\caption{Properties of DPM (dynamic VCG), BTM (dynamic AGV), and TU-GUM in the dynamic TU setting; adapted from \cite[Table~3]{CsokaEtAl2026}.}
\label{tab:dynamic-mechanism-properties}
\small
\renewcommand{\arraystretch}{1.13}
\setlength{\tabcolsep}{3pt}
\begin{tabularx}{\textwidth}{@{}>{\raggedright\arraybackslash}X *{3}{>{\centering\arraybackslash}p{0.16\textwidth}}@{}}
\toprule
Property & DPM & BTM & TU-GUM\\
\midrule
\multicolumn{4}{c}{\cellcolor{gray!15}\textit{Efficiency and strategic robustness}}\\
Efficiency is implemented in & \PEEmark & \PBEmark & \GUEmark\\
Efficiency survives IEWDS & \noprop & \noprop & \yesprop\\
All BNEs/PBEs are efficient & \noprop & \noprop & \yesprop\\
Collusion-proof & \noprop & \noprop & \yesprop\\
\midrule
\multicolumn{4}{c}{\cellcolor{gray!15}\textit{Other properties and extensions}}\\
Balanced budget & \noprop & \yesprop & \yesprop\\
Private actions are allowed & \noprop & \yesprop & \yesprop\\
Exit and re-entry are allowed & \yesprop & \noprop & \yesprop\\
Earlier types of other agents may be observed & \yesprop & \noprop & \yesprop\\
\bottomrule
\end{tabularx}
\end{table}
\FloatBarrier
\Needspace{8\baselineskip}
For two agents, however, AGV and its dynamic extension, the Balanced Team Mechanism (BTM), are also GUE implementations of efficiency, although their payment rules generally differ from TU-GUM.
More generally, every two-agent dynamic externality-payment mechanism is a GUE implementation; see Sections~4.2 and~4.4 of \cite{CsokaPongraczRodivilov2025}.

For two agents, if at most one agent receives a new private type in each period, TU-GUM coincides with BTM and therefore with AGV in the static case.
The two-period trading example of Athey and Segal has exactly this timing: the seller observes his type before the first period, whereas the buyer observes his type between the two periods \cite[Example~1]{AtheySegal2013}.
Consequently, BTM and TU-GUM coincide in that example \cite[Appendix~G]{CsokaEtAl2026}.

In a forthcoming paper, Cs\'oka and Viczi\'an introduce a variant of TU-GUM that coincides with AGV for two agents.
The main difference is that, among all GUE implementations of \(\chi\), the symmetrized TU-GUM minimizes the sum of the variances of the agents' utilities, whereas the variant minimizes the sum of the variances of the transfers.

\section{A three-agent social choice example}
\label{sec:three-agent-example}

There are three agents.
Independently,
\[
 \theta_i\in\{-6,10\}
 \qquad \Pr(\theta_i=-6)=\Pr(\theta_i=10)=1/2
\]
The decision is \(x\in\{0,1\}\), and
\[
 u_i(\theta_i,x)=x\theta_i
\]
The fixed efficient policy is
\[
 \chi(\widehat\theta)=
 \1\!\{\widehat\theta_1+\widehat\theta_2+\widehat\theta_3>0\}
\]
Thus \(x=0\) when at most one report equals 10, and \(x=1\) when at least two reports equal 10.
The efficient ex-ante expected total payoff is
\begin{equation}
 M=3/8 \cdot 14+1/8 \cdot 30=9
 \label{eq:M9}
\end{equation}
By symmetry, each agent's expected payoff under this policy is \(V_i(\varnothing)=3\).

Table~\ref{tab:mechanism-summary} summarizes the decision and transfer vectors for the symmetric VCG family, AGV, and symmetrized TU-GUM; permuting a report profile permutes the coordinates of the corresponding transfer vector.
\begin{table}[htbp]
\centering
\caption{Decisions and transfer vectors for the symmetric VCG family, AGV, and symmetrized TU-GUM.}
\label{tab:mechanism-summary}
\small
\renewcommand{\arraystretch}{1.12}
\begin{tabularx}{\textwidth}{@{}l *{4}{>{\centering\arraybackslash}X}@{}}
\toprule
Reports \(\widehat\theta_{\N}\)
 & \(\textcolor{cblue}{\tuplefit{-6}{-6}{-6}{-6}{-6}{-6}}\)
 & \(\textcolor{cblue}{\tuplefit{-6}{-6}{10}{-6}{-6}{10}}\)
 & \(\textcolor{cblue}{\tuplefitcompact{-6}{b+4}{b+4}{-6}{10}{10}}\)
 & \(\textcolor{cblue}{\tuplefit{10}{10}{10}{10}{10}{10}}\)\\
\midrule
Decision
 & \(\textcolor{dpurple}{x=0}\)
 & \(\textcolor{dpurple}{x=0}\)
 & \(\textcolor{dpurple}{x=1}\)
 & \(\textcolor{dpurple}{x=1}\)\\
\specialrule{\lightrulewidth}{0.35em}{0pt}
sym.\ VCG
 & \(\tuplefit{-6}{-6}{-6}{a}{a}{a}\)
 & \(\tuplefit{-6}{-6}{10}{b}{b}{a}\)
 & \(\tuplefitcompact{-6}{b+4}{b+4}{c}{b+4}{b+4}\)
 & \(\tuplefit{10}{10}{10}{c}{c}{c}\)\\
\addlinespace[0.10em]
AGV & \(\tuplefit{-6}{-6}{-6}{0}{0}{0}\) & \(\tuplefit{-6}{-6}{10}{-1}{-1}{2}\) & \(\tuplefitcompact{-6}{b+4}{b+4}{-2}{1}{1}\) & \(\tuplefit{10}{10}{10}{0}{0}{0}\)\\
\rowcolor{llgreen}
sym.\ TU-GUM
 & \(\tuplefit{-6}{-6}{-6}{0}{0}{0}\) & \(\tuplefit{-6}{-6}{10}{-3}{-3}{6}\) & \(\tuplefitcompact{-6}{b+4}{b+4}{2}{-1}{-1}\) & \(\tuplefit{10}{10}{10}{0}{0}{0}\)\\
\bottomrule
\end{tabularx}
\end{table}

The VCG row gives the full symmetric VCG family in this example.
Appendix~\ref{app:guvcg} discusses particular VCG rules and introduces GU-VCG.

\subsection{VCG calculation}

For a report profile \(\widehat\theta\), write
\[
 W_{-i}(\widehat\theta)
 :=\sum_{j\ne i}u_j\bigl(\widehat\theta_j,\chi(\widehat\theta)\bigr)
\]
Equation~\eqref{eq:VCGgeneral} gives
\[
 y_i^{\VCG,h}(\widehat\theta)
 =W_{-i}(\widehat\theta)+h_i(\widehat\theta_{-i})
\]
In this example, \(W_{-i}(\widehat\theta)=0\) when \(\chi(\widehat\theta)=0\), and
\(W_{-i}(\widehat\theta)=\sum_{j\ne i}\widehat\theta_j\) when \(\chi(\widehat\theta)=1\).
Consequently, for every agent \(i\),
\[
 y_i^{\VCG,h}(10,\widehat\theta_{-i})
 =y_i^{\VCG,h}(-6,\widehat\theta_{-i})+4
 \qquad
 \text{if }\widehat\theta_{-i}\in\{(-6,10),(10,-6)\},
\]
while
\[
 y_i^{\VCG,h}(10,\widehat\theta_{-i})
 =y_i^{\VCG,h}(-6,\widehat\theta_{-i})
 \qquad
 \text{if }\widehat\theta_{-i}\in\{(-6,-6),(10,10)\}.
\]
The VCG rules that are symmetric across agents are therefore exactly the three-parameter family displayed in Table~\ref{tab:mechanism-summary}.

\subsection{AGV calculation}

Unconditionally, each agent's expected payoff under this policy is 3, so the baseline expected payoff of the two agents other than \(i\) is 6.
If \(i\) reports \(-6\), the decision is 1 only when both other reports equal 10.
Their expected total payoff is therefore
\[
 1/4 \cdot (10+10)=5
 \qquad G_i(-6)=5-6=-1
\]
If \(i\) reports \(10\), the other agents' expected total payoff is
\[
 1/4 \cdot 0+1/4 \cdot 4+1/4 \cdot 4
 +1/4 \cdot 20=7
 \qquad G_i(10)=7-6=1
\]
Since \(n=3\), equation~\eqref{eq:AGV} gives
\[
 y_i^{\AGV}=G_i(\widehat\theta_i)
 -0.5\sum_{k\ne i}G_k(\widehat\theta_k)
\]
For instance,
\[
 y^{\AGV}(-6,-6,10)=
 \Bigl(-1-0.5\cdot(-1+1),
       -1-0.5\cdot(-1+1),
        1-0.5\cdot(-1-1)\Bigr)=(-1,-1,2)
\]

\subsection{TU-GUM: fixed-order and symmetrized versions}

We use the fixed order \(1,2,3\).
Thus
\[
 B_1=\varnothing,\qquad
 B_2=\{1\},\qquad
 B_3=\{1,2\}
\]

\medskip
\noindent\textbf{A direct demonstration of the bilateral payments.}\par
Consider the report profile \((-6,-6,10)\).
Before any report is processed, the anticipated payoff of each agent is
\[
 3/8 \cdot 10-1/8 \cdot 6=3
\]
First process \(\widehat\theta_1=-6\).
The anticipated payoff of each of agents 2 and 3 becomes
\[
 1/4 \cdot 10=2.5
\]
Hence \(\gamma^{1\to2}=\gamma^{1\to3}=-0.5\): agent 1 pays compensation of \(0.5\) to each of agents 2 and 3.

Next process \(\widehat\theta_2=-6\).
Agent 1's anticipated payoff rises from \(1/4 \cdot (-6)=-1.5\) to 0.
Thus \(\gamma^{2\to1}=1.5\), so agent 1 pays \(1.5\) to agent 2.
Agent 3's anticipated payoff falls from \(1/4 \cdot 10=2.5\) to 0.
Thus \(\gamma^{2\to3}=-2.5\), so agent 2 pays \(2.5\) to agent 3.

Finally, processing \(\widehat\theta_3=10\) does not change the decision or either earlier agent's anticipated payoff, so \(\gamma^{3\to1}=\gamma^{3\to2}=0\).
The resulting net transfers are therefore
\[
 y^{\GUM}(-6,-6,10)
 =(-0.5-0.5-1.5,\;0.5+1.5-2.5,\;0.5+2.5)
 =(-2.5,-0.5,3)
\]
This is the simple meaning of the rule: whenever a report changes another agent's anticipated payoff, the reporting agent is paid that change by the affected agent.

The following table gives all six sequential bilateral externalities for every report profile, obtained directly from \eqref{eq:gammaB}.

\begin{table}[htbp]
\centering
\caption{Sequential bilateral externalities for the fixed order \(1,2,3\).}
\label{tab:sequential-externalities}
\small
\begin{tabular*}{\textwidth}{@{\extracolsep{\fill}}c c c c c c c@{}}
\toprule
Reports & \(\gamma^{1\to2}\) & \(\gamma^{1\to3}\)
& \(\gamma^{2\to1}\) & \(\gamma^{2\to3}\)
& \(\gamma^{3\to1}\) & \(\gamma^{3\to2}\)\\
\midrule
\(\reporttriplet{-6}{-6}{-6}\) & \(\alignednum{-0}{.5}\) & \(\alignednum{-0}{.5}\) & \(\alignednum{1}{.5}\) & \(\alignednum{-2}{.5}\) & \(\alignednum{0}{}\) & \(\alignednum{0}{}\)\\
\(\reporttriplet{-6}{-6}{10}\) & \(\alignednum{-0}{.5}\) & \(\alignednum{-0}{.5}\) & \(\alignednum{1}{.5}\) & \(\alignednum{-2}{.5}\) & \(\alignednum{0}{}\) & \(\alignednum{0}{}\)\\
\(\reporttriplet{-6}{10}{-6}\) & \(\alignednum{-0}{.5}\) & \(\alignednum{-0}{.5}\) & \(\alignednum{-1}{.5}\) & \(\alignednum{2}{.5}\) & \(\alignednum{3}{}\) & \(\alignednum{-5}{}\)\\
\(\reporttriplet{-6}{10}{10}\) & \(\alignednum{-0}{.5}\) & \(\alignednum{-0}{.5}\) & \(\alignednum{-1}{.5}\) & \(\alignednum{2}{.5}\) & \(\alignednum{-3}{}\) & \(\alignednum{5}{}\)\\
\(\reporttriplet{10}{-6}{-6}\) & \(\alignednum{0}{.5}\) & \(\alignednum{0}{.5}\) & \(\alignednum{-2}{.5}\) & \(\alignednum{1}{.5}\) & \(\alignednum{-5}{}\) & \(\alignednum{3}{}\)\\
\(\reporttriplet{10}{-6}{10}\) & \(\alignednum{0}{.5}\) & \(\alignednum{0}{.5}\) & \(\alignednum{-2}{.5}\) & \(\alignednum{1}{.5}\) & \(\alignednum{5}{}\) & \(\alignednum{-3}{}\)\\
\(\reporttriplet{10}{10}{-6}\) & \(\alignednum{0}{.5}\) & \(\alignednum{0}{.5}\) & \(\alignednum{2}{.5}\) & \(\alignednum{-1}{.5}\) & \(\alignednum{0}{}\) & \(\alignednum{0}{}\)\\
\(\reporttriplet{10}{10}{10}\) & \(\alignednum{0}{.5}\) & \(\alignednum{0}{.5}\) & \(\alignednum{2}{.5}\) & \(\alignednum{-1}{.5}\) & \(\alignednum{0}{}\) & \(\alignednum{0}{}\)\\
\bottomrule
\end{tabular*}
\end{table}

Taking the net bilateral payments in \eqref{eq:seqGUM} gives the following.
\begin{table}[htbp]
\centering
\caption{Decisions and transfer vectors under fixed-order TU-GUM.}
\label{tab:fixed-order-transfers}
\small
\setlength{\tabcolsep}{3pt}
\begin{tabular*}{\textwidth}{@{\extracolsep{\fill}}c c c c c c@{}}
\toprule
Reports & Decision & \(y^{\GUM}\)
& Reports & Decision & \(y^{\GUM}\)\\
\midrule
\(\reporttriplet{-6}{-6}{-6}\) & 0 & \(\lefttransfertriplet{-2.5}{-0.5}{3}\)
& \(\reporttripletpos{10}{-6}{-6}\) & 0 & \(\righttransfertriplet{8.5}{-4.5}{-4}\)\\
\(\reporttriplet{-6}{-6}{10}\) & 0 & \(\lefttransfertriplet{-2.5}{-0.5}{3}\)
& \(\reporttripletpos{10}{-6}{10}\) & 1 & \(\righttransfertriplet{-1.5}{1.5}{0}\)\\
\(\reporttriplet{-6}{10}{-6}\) & 0 & \(\lefttransfertriplet{-2.5}{6.5}{-4}\)
& \(\reporttripletpos{10}{10}{-6}\) & 1 & \(\righttransfertriplet{-1.5}{0.5}{1}\)\\
\(\reporttriplet{-6}{10}{10}\) & 1 & \(\lefttransfertriplet{3.5}{-3.5}{0}\)
& \(\reporttripletpos{10}{10}{10}\) & 1 & \(\righttransfertriplet{-1.5}{0.5}{1}\)\\
\bottomrule
\end{tabular*}
\end{table}
\FloatBarrier
Every displayed transfer vector sums to zero.
The fixed order makes the individual transfers asymmetric, but, for every fixed report profile of the other agents, averaging a truthful agent's utility over his own two equiprobable types in Table~\ref{tab:fixed-order-transfers} gives \(3\).
Because types are independent and reports are simultaneous, truthful reporting therefore guarantees ex-ante expected utility \(3\) against every strategy profile of the other agents in this game.

Averaging the TU-GUM payment rule over all orders of the three agents gives symmetrized TU-GUM.
Collecting the calculations and using symmetry gives its row in Table~\ref{tab:mechanism-summary}.

\subsection{Utility guarantees and their consequences}

Table~\ref{tab:mechanism-summary} makes the symmetrized TU-GUM guarantee especially transparent.
Up to permutation of the two other reports, the utility of a truthful agent averaged over his own type is
\begin{align*}
 (10,10):&\quad \frac{(-6+2)+(10+0)}2=3\\
 (-6,10):&\quad \frac{(0-3)+(10-1)}2=3\\
 (-6,-6):&\quad \frac{(0+0)+(0+6)}2=3
\end{align*}
The same guarantee was already verified above for the fixed-order TU-GUM.
Hence both fixed-order and symmetrized TU-GUM implement the efficient outcome in GUE.
Moreover, under truthful reporting, \(\chi\) selects the efficient decision and transfers balance, so the agents' expected total utility is the maximal expected total payoff, \(M=9\), by \eqref{eq:M9}.
Thus the three guarantees exhaust feasibility:
\[
 M=9=3+3+3
\]

This calculation also shows how TU-GUM creates these guarantees.
At the level of externality accounting, VCG and AGV build on Vickrey's idea that the reporting agent should internalize the externality caused by his report.
TU-GUM adds exact compensation of the agent affected by each externality.
Sequential bilateral settlement therefore both gives the reporting agent the appropriate incentive and protects the affected agent's anticipated utility, making GUE implementation possible.

For comparison, Table~\ref{tab:conditional-utilities} gives the corresponding own-type averages under AGV and TU-GUM.
\begin{table}[htbp]
\centering
\caption{Utility of a truthful agent, averaged over his own type, for fixed reports of the others.}
\label{tab:conditional-utilities}
\small
\begin{tabular*}{\textwidth}{@{\extracolsep{\fill}}c c c c@{}}
\toprule
Mechanism & \((10,10)\) & \((-6,10)\) or \((10,-6)\) & \((-6,-6)\)\\
\midrule
AGV & 1 & 5 & 1\\
TU-GUM & 3 & 3 & 3\\
\bottomrule
\end{tabular*}
\end{table}
\FloatBarrier
Both displayed rules give each truthful agent ex-ante expected utility 3, but only TU-GUM makes this level a guarantee against arbitrary behavior of the other agents.
AGV and both versions of TU-GUM are budget-balanced.

\paragraph{Coalition deviations in the example.}
The same transfer table makes the contrast between AGV and TU-GUM concrete.
Under AGV, if agents 1 and 2 have types \((-6,10)\), they can jointly report \((10,10)\). Their aggregate utility then rises from \(1\) to \(6\) when agent 3 reports \(-6\), and from \(3\) to \(4\) when agent 3 reports \(10\); transferable side payments can make both coalition members better off.\footnote{An even stronger example appears in \cite[Section~2]{CsokaEtAl2026}: both coalition members are strictly better off without side transfers, and the resulting strategy profile is both a strict Nash equilibrium and a super strong Nash equilibrium, despite inducing a socially inefficient outcome.}
In contrast, under either version of TU-GUM, no coalition can raise the expected total utility of its members by a joint report deviation when the outsiders remain truthful: the outsiders retain their guarantees, while expected total utility is already maximal.
This is the transferable-utility collusion-proofness implied by the GUE property.

\section{Generalizations of TU-GUM}

The sequential-externality construction extends beyond the static social choice environment \cite{CsokaEtAl2026}:
\begin{itemize}
 \item \textbf{Dynamic games.}
 The sequential bilateral externality of a report includes its effect on the other agent's anticipated utilities in future periods as well as in the current period.
 \item \textbf{Private actions.}
 The efficient policy may recommend private actions in addition to public decisions; continuation values and sequential externalities are calculated under this enlarged policy.
 \item \textbf{Contractible interdependencies between working processes.}
 An agent's working process may generate contractible consequences that directly affect other agents' working processes and the outcome of the project.
 The efficient continuation policy and the sequential externalities may condition on these consequences even when the underlying actions, chance events, and costs remain private \cite{Csoka2015}.
\end{itemize}
All three extensions still yield exact GUE implementations in the TU setting when the participating set and initial type profile are fixed as part of the mechanism specification.
Two further directions are worth highlighting:
\begin{itemize}
 \item \textbf{Transfer-free mechanisms.}
 A variant of TU-GUM also applies without monetary transfers, giving an \(\varepsilon\)-GUE implementation with vanishing per-period error \cite{CsokaEtAl2026}.
 \item \textbf{Non-fixed initial types.}
 There are two prior-free extensions beyond the assumption that the initial type profile is fixed as part of the mechanism specification.
 One is transfer-free \cite{CsokaPriorFree2025}; the other introduces an auction stage in which privately informed agents submit contingent offers \cite{Csoka2015}.
 The former gives particularly clean and mathematically tractable results, while the latter may be more important for practical applications and leads to the Project Management Model and Mechanism discussed next.
\end{itemize}

\section{The Project Management Model and Mechanism}
\label{sec:pmm}

The paper \emph{\href{https://arxiv.org/abs/cs/0602009}{Efficient Teamwork}} \cite{Csoka2015} introduced a still more general, prior-free model with adverse selection. Each agent privately observes an initial type that specifies his stochastic dynamic working process. The planner knows neither these types nor any prior distribution over them, and the model imposes no independence assumption on the agents' initial types. The goal is to select the best team for the project and get its members to work together efficiently.

A mechanism defined by a single sentence offers a reasonably effective solution to this general problem. It combines the GUE technique with the arguably attractive and robust properties of first-price tendering for team selection.
We next show how TU-GUM arises as a very special case of the general mechanism.

\subsection{The simplified Project Management Model}

The simplified Project Management Model has the following elements.
\begin{itemize}
 \item \textbf{Private work processes.}
 Each candidate agent \(i\in\N\) privately observes an initial type \(\theta_i\).
 This type specifies a stochastic dynamic working process, including feasible private actions, privately observed chance events and their probabilities, and the resulting contractible output \(r_i\in\mathcal R_i\).
 At each chance node, the next state is generated according to the probabilities specified by that working process.
 Agent \(i\)'s complete realized private work history is denoted by \(\xi_i\).
 Time is absolute: every potential decision opportunity and chance event, across all agents, has a fixed position on a common timeline.

 \item \textbf{Team selection.}
 The principal selects a team \(A\subseteq\N\); its members are called the accepted agents.
 Their working processes then run on the common time scale.
 Candidates outside \(A\) are called the rejected agents; they receive zero utility and exit the game.

 \item \textbf{Contractible results and social choice.}
 By the end of the project, a contractible result vector
 \[
   r=(r_0,r_A)
   \in \mathcal R_0\times\prod_{i\in A}\mathcal R_i
 \]
 is available, where \(r_0\) is the social choice selected by the mechanism and
 \(r_A=(r_i)_{i\in A}\) collects the contractible results associated with the accepted agents' working processes.

 \item \textbf{Payoffs.}
 For a selected team \(A\) and result vector \(r=(r_0,r_A)\), the payoff of agent \(i\in A\) is
 \[
   u_i(\xi_i,r)
 \]
 while the principal's payoff is
 \[
   u_0(r)
 \]
 The familiar private-cost formulation is the special case
 \(u_i(\xi_i,r)=-c_i(\xi_i)\).

 \item \textbf{Communication and transfers.}
 Communication is available throughout the project.
 Messages are contractible, although the information they report need not be verifiable.
 If \(m\) is the complete communication history, the signed transfer from the principal to agent \(i\in A\) may be any prescribed function \(t_i(m,r)\).

 \item \textbf{Utilities.}
 Write \(y_i=t_i(m,r)\) and \(y_0=-\sum_{i\in A}y_i\).
 Each accepted agent \(i\) has a Bernoulli utility function
 \[
   \bar u_i(\xi_i,r,y_i)
 \]
 and the principal has a Bernoulli utility function
 \[
   \bar u_0(r,y_0)
 \]
 Under a contingent plan, their utilities are
 \[
   U_i=\mathbb E\!\left[\bar u_i(\xi_i,r,y_i)\right]
   \quad(i\in A)
   \qquad
   U_0=\mathbb E\!\left[\bar u_0(r,y_0)\right]
 \]
 where the expectations are taken over the chance events and any randomization induced by the contingent plan and the stochastic working processes.
 The Bernoulli utility functions may be nonquasilinear, but the framework is intended for settings in which monetary transfers remain a meaningful means of compensation.

 \item \textbf{Quasilinear benchmark.}
 Under quasilinearity,
 \[
   U_i
   =\mathbb E\!\left[u_i(\xi_i,r)+y_i\right]
   \quad(i\in A)
   \qquad
   U_0
   =\mathbb E\!\left[u_0(r)+y_0\right]
 \]
 Hence total utility equals expected total payoff:
 \[
   U_0+\sum_{i\in A}U_i
   =\mathbb E\!\left[u_0(r)+\sum_{i\in A}u_i(\xi_i,r)\right]
 \]

 \item \textbf{Objectives.}
 Depending on the application, the objective may be to maximize the principal's utility or project efficiency.
 We focus primarily on settings in which failures of cooperation during execution can cause substantial losses, whereas sufficiently competitive tendering is expected to leave relatively little room for improvement through further optimization of the auction stage.

 \item \textbf{Idealized assumptions and robustness.}
 The idealized version uses the quasilinear benchmark above and assumes perfect informational separation: each player's newly realized private chance event is independent of all information previously available to the other players.
 The Project Management Mechanism and its guaranteed-utility logic apply to arbitrary vNM utility functions, allowing, in particular, risk aversion, and also accommodate limited departures from perfect informational separation.
 The explicit sell-the-firm and TU-GUM representations below use the quasilinear benchmark.

 \item \textbf{Further features of the full model.}
 First, the principal may also have her own work process.
 Second, agents may produce contractible results at any time during the project and may be affected by results produced by others. This accommodates, for example, precedence constraints between different agents' subtasks and the sharing of common resources.
\end{itemize}
For simplicity, we omit these additional features from the formal analysis below.

\subsection{Definition of the Project Management Mechanism}

\paragraph{Contract offers.}
Each candidate \(i\) submits a contract offer specifying, for every complete communication history \(m\) and result vector \(r\), a signed transfer \(t_i(m,r)\) from the principal to agent \(i\). If the principal accepts agent \(i\), this transfer rule is binding.

\paragraph{Mechanism.}
Given the offers, consider the ``worst-case game'' in which each accepted agent \(i\) can choose any communication strategy and produce any result \(r_i\in\mathcal R_i\), and the accepted agents may coordinate arbitrarily.
The principal evaluates her strategies in this worst-case game.
The mechanism is then the following.
\begin{center}
\textcolor{red}{\textbf{The principal chooses a strategy that maximizes her minimum utility.}}
\end{center}
Together with the contract protocol above, this highlighted sentence defines the \emph{Project Management Mechanism} (PMM), in its first-price form.
Importantly, the PMM is prior-free.

\subsection{Interpretation and rationale}

\paragraph{The principal's choice.}
Effectively, the principal requires each accepted agent to bear all the risks arising from his own uncertainty, on terms specified by his offer.
This allows her to secure a guaranteed utility level and makes different subsets of offers comparable.
The maximin rule is the precise formulation of this requirement.
(It should not be confused with maximin preferences.)

\paragraph{Truthful offers.}
A contract offer is \emph{truthful at utility level \(c_i\)} if it gives agent \(i\) a truthful strategy guaranteeing utility at least \(c_i\) against every strategy profile of the principal and the other agents.
Informally, it is equivalent to the agent saying the following:
\begin{quote}
``Here is my dynamic stochastic working process, including all probabilities, decision opportunities, costs, and preferences.
I accept any contingent plan and payment rule under which, by being truthful, I can guarantee myself at least my specified utility level, no matter what the others do.''
\end{quote}
In property-rights terms, the agent specifies a utility level and offers the principal all indirect control over his working process that is compatible with guaranteeing that level through truthful cooperation.
Such offers allow the principal to assign to each agent the risks generated by that agent's own uncertainty without creating a strategic opportunity for the agent.
The principal can thereby secure a utility guarantee independent of how the project unfolds.
This is the GUE-type logic behind the Project Management Mechanism.
Under the idealized assumptions, conditional on the accepted team and the agents' protected utility levels, truthful offers allow the principal to direct the team's joint execution efficiently and secure the remaining surplus as a fixed utility level.
Because a truthful offer reflects the agent's actual preferences, it can incorporate nonquasilinearity, including risk aversion.

\section{Specializations and the full Project Management Model}

We can gain a clearer understanding of how the Project Management Mechanism works and of its properties by examining simpler special cases, each of which makes a particular feature of the mechanism more transparent.

We use four binary labels. In each pair, green marks the extension and red the corresponding simpler case.
\begin{itemize}
 \item \featureon{D}{Dynamic execution.}
 The interaction has multiple stages.
 Execution may involve stochastic working processes, private actions, and new private chance events.
 \featureoff{D}{Static allocation game.}
 The interaction consists only of allocating fixed objects or tasks and making monetary transfers.
 There are no execution-stage private actions or chance events.

 \item \featureon{M}{Multi-agent.}
 Several agents participate in the project and their tasks may interact.
 \featureoff{M}{Single-agent execution.}
 One agent participates in the execution.

 \item \featureon{R}{Robust.}
 The utility and information assumptions may depart from quasilinearity and perfect informational separation as described above.
 \featureoff{R}{Idealized.}
 The idealized assumptions are imposed.

 \item \featureon{A}{Auction.}
 Privately informed candidates submit offers and the principal selects the participants, as in the mechanism defined above.
 With \featureoffmark{M}, several candidates may compete, but at most one is accepted.
 \featureoff{A}{No auction.}
 The tendering and acceptance stage is suppressed: the accepted agents and their truthful contingent offers are already given, and we start the game from this point.
\end{itemize}
Thus \featureoffmark{M}\featureoffmark{A} means that the game has one principal and one agent.
We discuss three informative specializations rather than all sixteen combinations, and then return to the full Project Management Model.

\subsection{Single-contractor first-price tendering}
\label{sec:single-contractor-tendering}

\begin{center}
\featureon{D}{Dynamic execution} \qquad \featureoff{M}{Single-agent} \qquad \featureon{R}{Robust} \qquad \featureon{A}{Auction}
\end{center}

Consider the Project Management Mechanism in the special case of the model in which at most one candidate is accepted and participates in the execution of the project, although several privately informed candidates may compete for this position.
The dynamic structure of the general model is retained: each candidate's private initial type may specify a stochastic working process with privately observed chance events, private actions, and contractible results.
Under the PMM, candidates submit contract offers specifying contingent payment rules.
The principal's maximin utility conditional on accepting an offer is called the \emph{value} of that offer, and she accepts an offer of highest value.

\paragraph{One-parameter truthful offers with a known principal type.}
Suppose first that the principal's type is common knowledge when offers are submitted.
Hence each candidate can determine the value of any offer.
For a given candidate, strategies that make offers of the same value are accepted with the same probability.
Among them, it can be shown that the recommended truthful strategy gives the candidate the highest utility conditional on acceptance.
The recommended strategies can therefore be indexed by the value of the offer.
This reduction relies on both the single-contractor assumption and the candidate's knowledge of the principal's type.
With quasilinear utilities, these recommended strategies maximize the sum of the offer's value and the candidate's utility conditional on acceptance, and transfers allow this maximum to be divided freely between them.

\paragraph{The quasilinear sell-the-firm representation.}
Let the quasilinear principal obtain payoff \(u_0(r)\) from the contractible result \(r\), while the contractor incurs realized cost \(\kappa(\xi)\).
For a scalar bid \(b\), consider the contract offer
\[
  t_b(r)=u_0(r)-b
\]
The principal's realized utility is
\[
  u_0(r)-t_b(r)=b
\]
while the contractor's realized utility is
\[
  -\kappa(\xi)+t_b(r)=u_0(r)-\kappa(\xi)-b
\]
Thus the contractor is the residual claimant of the project's monetary surplus and, conditional on being selected, his private decisions maximize its expected value.
The sum of the value of an offer and the candidate's utility conditional on acceptance cannot exceed the maximum expected total payoff attainable by the principal and that candidate.
The sell-the-firm strategies attain this bound, and varying \(b\) yields every division of it.

For nonquasilinear vNM utilities, the corresponding truthful offers need not have this simple affine form, but the one-parameter reduction does not rely on quasilinearity.

\paragraph{First-price tendering.}
With a commonly known principal type, each candidate can restrict attention to the one-parameter family above, and the selected candidate is held to his own offer.

In the quasilinear benchmark, because the principal receives \(b\) in every realization under \(t_b\), the mechanism selects the highest bid.
The tender therefore reduces exactly to a single-item first-price auction in the bids \(b\)---equivalently, after reversing signs, to a procurement auction.
With quasilinear bidders, the corresponding second-price auction gives the familiar efficient comparison.
The first-price reduction is an exact description of the induced tendering game, not an efficiency result: first-price tendering does not implement efficient contractor selection in general.

\paragraph{Consortium aggregation.}
Against fixed offers of the outsiders, only a coalition's most competitive offer can determine whether one of its members is selected and what the principal is guaranteed.
The coalition can reproduce the same outcome by entering openly as one consortium and assigning the project internally.
Thus separate identities give the coalition no advantage over an explicit consortium.
This does not prevent consortium formation; it removes an additional gain from concealing the consortium.

\paragraph{Privately known principal type.}
If the principal's type is not common knowledge when offers are submitted, the Project Management Mechanism itself remains unchanged, but the one-parameter reduction no longer applies.
Which concessions are most attractive to the principal for a given division of the available value can depend on her type, so richer contingent offers may then be relevant.

\subsection{The TU-GUM specialization}
\label{sec:tugumspecialization}

\begin{center}
\featureon{D}{Dynamic execution} \qquad \featureon{M}{Multi-agent} \qquad \featureoff{R}{Idealized} \qquad \featureoff{A}{No auction}
\end{center}

The \featureoffmark{A} convention above means that the accepted agents and their truthful contingent offers are already given, and we start the game from this point.
The dynamic and multi-agent structure is retained: the accepted agents may have interacting stochastic working processes, private actions, and newly arriving private information during execution.

The exact TU-GUM reduction imposes the idealized assumptions in two respects:
\begin{itemize}
 \item \textbf{Utilities are quasilinear.}
 In particular, the agents are risk-neutral.
 \item \textbf{No robustness to informational dependence.}
 Each player's newly realized private chance event is assumed to be independent of all information previously available to the other players.
 Unlike the general mechanism, this reduction is not robust to even small departures from this assumption.
\end{itemize}

Let \(c_i\) be the protected utility level specified by accepted agent \(i\)'s truthful offer, and let \(W^*\) be the maximum feasible expected total utility given the set of accepted agents \(A\).
By definition, agent \(i\) has a truthful strategy guaranteeing \(c_i\) against every strategy profile of the principal and the other agents.
If all accepted agents use these strategies, no strategy of the principal can guarantee her more than \(W^*-\sum_{i\in A}c_i\).
Conversely, the fixed-order TU-GUM construction \cite{CsokaEtAl2026} gives the principal a strategy guaranteeing exactly this residual and, together with the agents' truthful strategies, implements an efficient continuation.
The guaranteed utilities of the principal and the accepted agents therefore sum to \(W^*\), so these strategies form a GUE of the continuation game; in particular, the principal's strategy inducing TU-GUM is maximin.

Applying Myerson's Revelation Principle \cite{Myerson1979} to the resulting continuation equilibrium replaces the potentially rich communication protocol admitted by the truthful offers with an outcome-equivalent direct mechanism in which the agents report only their private information as it arrives.
In \cite{Csoka2015}, the principal's maximin strategy is calculated by backward recursion along the common timeline.
When all chance events occur at distinct time points, the revelation-mechanism transform of this construction is exactly the corresponding fixed-order TU-GUM rule \eqref{eq:seqGUM}; the chronological order of the events determines the order of the rule, and report-independent constants determine the protected utility levels.
Appendix~\ref{app:pmmstatic} works out the construction explicitly in the three-agent static example.

\subsection{The first-price combinatorial auction specialization}

\begin{center}
\featureoff{D}{Static allocation game} \qquad \featureon{M}{Multi-agent} \qquad \featureon{R}{Robust} \qquad \featureon{A}{Auction}
\end{center}

With fixed objects or tasks, several potential participants, and a nonadditive value of the selected collection, the Project Management Mechanism specializes to a first-price combinatorial auction---or, after reversing signs, to combinatorial procurement.

The first-price format remains meaningful without quasilinearity.
Risk aversion may change bidding behavior and performance, but it does not invalidate the first-price competitive rationale.
Since under \featureoffmark{D} there are no execution-stage private chance events, the informational-separation component of \featureonmark{R} is vacuous.
In this specialization, \featureonmark{R} refers only to departures from quasilinearity.

The consortium-aggregation argument also extends from the single-item case: against fixed outside bids, a coalition cannot obtain more by submitting separate combinatorial bids under a coordinated strategy than by entering openly as a single consortium with an internal allocation and payment rule.
Consortium formation itself may still benefit its members; the point is only that concealed separate entry provides no additional advantage.
Intuitively, provided sufficient outside competition remains, the formation of relatively small consortia need not fundamentally undermine the competitive performance of first-price tendering.
With quasilinear bidders, the corresponding VCG mechanism provides the familiar efficient benchmark but is highly vulnerable to collusion.

\subsection{The full Project Management Model}

\begin{center}
\featureon{D}{Dynamic execution} \qquad \featureon{M}{Multi-agent} \qquad \featureon{R}{Robust} \qquad \featureon{A}{Auction}
\end{center}

The preceding constructions are specializations of the same one-sentence mechanism.
The full Project Management Model switches on all four extensions: privately informed candidates submit contingent first-price offers, the principal selects a team, and the accepted offers induce a potentially rich dynamic communication and execution game in which she chooses a maximin strategy.
This is the setting studied in \emph{Efficient Teamwork} \cite{Csoka2015}.

Its natural static benchmark is the first-price combinatorial auction.
Even for that benchmark, the available general results do not provide useful formal guarantees matching the levels of allocative efficiency and expected revenue that one would expect under even moderately strong competition.
The familiar second-price alternative has solid mathematical proofs but is not unambiguously superior: it implements efficiency\footnote{For combinatorial auctions, it implements efficiency in DSE. For the general Project Management Model, it implements efficiency in NGUE and therefore also in PBE.} but is highly vulnerable to collusion.
The first-price mechanism instead has the weaker consortium-aggregation property described above, and the usual competitive intuition suggests that it should perform reasonably well under moderate competition, although this intuition is difficult to capture in useful general theorems.

The same competitive intuition suggests that, even in this richer setting, the losses under a reasonable degree of competition should be essentially of the same kind as those already familiar from first-price combinatorial auctions, rather than being compounded by additional losses from failures of cooperation during execution.
For the Project Management Mechanism, \cite{Csoka2015} establishes exact efficiency results in two special cases: under perfect competition, and when the agents know one another's initial types while these types remain unknown to the principal.

\clearpage
\appendix

\section{Weak-dominance fragility of AGV}
\label{app:agvfragility}

In this appendix, type spaces are finite, types are independent with full support, and randomized reporting strategies are allowed in dominance comparisons.
We first give a simple example in which exhaustive IEWDS eliminates every efficient strategy profile for a fixed efficient decision policy.
After recording the scope of the unique-decision protection result, we end with a stronger example in which the failure occurs for every deterministic efficient decision policy.

\subsection{A three-agent counterexample}

Agents 1 and 2 have independent equiprobable types \(L,H\); agent 3 has a single type.
The public decision is \(N\) (none), \(S\) (small investment), or \(B\) (big investment):
\[
\begin{array}{c|ccc}
 & \text{Agent 1} & \text{Agent 2} & \text{Agent 3}\\
 & L\ \text{or}\ H & L\ \text{or}\ H & \\ \hline
N & 0 & 0 & 0\\
S & 10\ \text{or}\ 20 & 0 & -10\\
B & 20\ \text{or}\ 32 & 5\ \text{or}\ 10 & -30
\end{array}
\]
Fix the following efficient decision policy:
\[
\begin{array}{c|cc}
 & \theta_2=L & \theta_2=H\\ \hline
\theta_1=L & N & N\\
\theta_1=H & S & B
\end{array}
\]
There are ties: at \((L,L)\) both \(N\) and \(S\) are efficient, while at \((L,H)\) all three decisions are efficient.
Thus efficiency does not determine a unique decision policy.

After report-independent normalizations, the AGV incentive terms are
\[
G_1(L)=0,\qquad G_1(H)=-15,\qquad
G_2(L)=0,\qquad G_2(H)=-4,\qquad G_3=0
\]
Hence changing a report from \(L\) to \(H\) costs agent 1 exactly \(15\) transfer units and agent 2 exactly \(4\).

To allow randomized reporting, write, for \(i=1,2\),
\[
a_i=\Pr(\widehat\theta_i=H\mid \theta_i=L),\qquad
b_i=\Pr(\widehat\theta_i=H\mid \theta_i=H),\qquad
q_i=\frac{a_i+b_i}{2}
\]
The net gain from reporting \(H\) rather than \(L\), conditional on the true type, is
\[
\begin{array}{c|cc}
 & \theta_i=L & \theta_i=H\\ \hline
\text{Agent 1} & -5+10q_2 & 5+12q_2\\
\text{Agent 2} & -4+5q_1 & -4+10q_1
\end{array}
\]
Consequently, every complete IEWDS procedure is forced through the following implications; arbitrary other deletions may be interspersed:
\[
b_1=1\ \Longrightarrow\ q_1\ge\tfrac12
\ \Longrightarrow\ b_2=1\ \Longrightarrow\ q_2\ge\tfrac12
\ \Longrightarrow\ a_1=1\ \Longrightarrow\ q_1=1
\ \Longrightarrow\ a_2=1
\]
Here \(b_1=1\), \(b_2=1\), and \(a_2=1\) are forced by strict dominance at the relevant reduced strategy sets; \(a_1=1\) is forced by weak dominance, with strict inequality against a surviving strategy with \(q_2>1/2\).
The intervening statements about \(q_1\) and \(q_2\) follow immediately from their definitions.

\paragraph{Order independence.}
The ``always report \(H\)'' strategy of either agent cannot be the first of the two such strategies to be deleted: against the other agent's always-\(H\) strategy, it is strictly optimal for both of that player's types.
Hence these two strategies survive every intermediate reduction and provide the strict witnesses needed above.
The unique survivor is therefore that agents 1 and 2 always report \(H\).
The mechanism always selects \(B\), which is inefficient at \((L,L)\) because \(B\) has total payoff \(-5\), whereas \(N\) and \(S\) have total payoff \(0\).

Thus exhaustive IEWDS under AGV eliminates every efficient strategy profile.

\subsection{Unique-decision protection: scope and limits}

\begin{proposition}[Unique-decision protection]
Suppose the efficient decision is unique at every type profile.
Then the truthful reporting profile survives every iterated elimination of weakly dominated strategies in the AGV mechanism, even when randomized reports are allowed.
\end{proposition}

\begin{proof}
Fix agent \(i\) and a realized type \(\theta_i\).
Suppose all other agents report truthfully.
Up to a constant independent of \(i\)'s report, agent \(i\)'s interim expected utility from report \(r_i\) is
\[
\E_{\theta_{-i}}\Bigl[
\sum_{j\in\N}u_j\bigl(\theta_j,\chi(r_i,\theta_{-i})\bigr)
\Bigr]
\]
This is the standard AGV welfare identity: the incentive term makes agent \(i\) internalize the others' expected payoff, while the redistribution term involving the other agents' \(G_j\)'s is independent of \(r_i\).

Truthful reporting selects the efficient decision for every \(\theta_{-i}\).
If another pure report changes the decision at some \(\theta_{-i}\), uniqueness of the efficient decision makes total welfare strictly smaller there; full support makes this event have positive probability.
Hence that report is strictly worse against truthful opponents.
If a pure report never changes the decision against truthful opponents, then, because every opponent report profile is a possible type profile, it never changes the decision against any opponent reports and is payoff-equivalent to truth everywhere.
Hence any randomized report is either strictly worse than truth against truthful opponents or payoff-equivalent to truth everywhere, and therefore cannot weakly dominate it.

Now consider any IEWDS sequence and suppose, for contradiction, that some truthful type-report is the first truthful one deleted.
Immediately before that deletion, every opponent still has all of her truthful type-reports available.
The truthful opponent profile is therefore still a feasible witness against the alleged dominating strategy, contradicting the previous paragraph.
Thus no truthful type-report can ever be deleted, and the truthful efficient profile survives every elimination order.
\end{proof}

Unique efficient decisions remove the sharp IEWDS diagnosis but do not by themselves resolve the underlying equilibrium-selection problem.
An arbitrarily small payoff perturbation can make the efficient decision unique at every type profile while causing one weak-dominance step in the argument to fail by an arbitrarily small amount.
The perturbed game remains arbitrarily close to the counterexample, so truthful play does not thereby become intuitively compelling.

The dynamic extension of AGV is more fragile: the same failure can occur even in generic environments with a unique efficient decision policy.

\begin{proposition}[{\cite{Csoka2021,CsokaEtAl2026}}]
For the Balanced Team Mechanism (dynamic AGV; \cite{AtheySegal2013}), there exists a generic environment with at least three agents and a unique efficient decision policy in which exhaustive IEWDS eliminates all efficient strategy profiles.
\end{proposition}

\subsection{A counterexample for every deterministic efficient policy}

The following stronger example was found a few days ago by Alexander Piankov.

\begin{proposition}[Every deterministic efficient policy]
There is a three-agent finite independent-private-values environment with two public decisions and exactly two deterministic efficient decision policies such that, under either policy, every exhaustive IEWDS sequence eliminates all efficient strategy profiles, independently of the elimination order, even when randomized reports may be used in dominance comparisons.
\end{proposition}

\begin{proof}
Let
\[
\Theta_1=\{\alpha,\beta,\gamma\},\qquad
\Theta_2=\{\ell,h\},\qquad
|\Theta_3|=1
\]
with uniform independent priors, and let the two decisions be \(y,z\).
The decision utilities are
\[
\begin{array}{c|ccc|cc|c}
 &u_1(\alpha)&u_1(\beta)&u_1(\gamma)
 &u_2(\ell)&u_2(h)&u_3\\ \hline
 y&0&2&3&2&5&0\\
 z&0&0&0&0&0&4
\end{array}
\]
The efficient correspondence is
\[
\begin{array}{c|cc}
 &\ell&h\\ \hline
\alpha&z&y\\
\beta&\{y,z\}&y\\
\gamma&y&y
\end{array}
\]
Thus there is one efficient tie, at \((\beta,\ell)\), and exactly two deterministic efficient policies, denoted by \(\chi^Y\) and \(\chi^Z\) according to their choice at that tie.

When an agent compares two own reports, all report-independent transfer terms cancel.
After payoff-equivalent reports have been identified, if the payoff difference between the two remaining reports takes both signs against surviving opponent report rules, neither report can be dominated, even by a lottery over them.

Under \(\chi^Y\), reports \(\beta\) and \(\gamma\) of agent 1 are payoff-equivalent; write \(n\) for this report class.
Let \(q=\Pr(\widehat\theta_2=h)\) and \(p_\alpha=\Pr(\widehat\theta_1=\alpha)\), where these probabilities combine the prior and the possibly randomized reporting rules.
Up to report-independent constants, the relevant payoff differences are
\[
\begin{array}{c|c@{\qquad}c|c}
\text{type 1}&U_1(\alpha)-U_1(n)
&\text{type 2}&U_2(h)-U_2(\ell)\\ \hline
\alpha&1&\ell&2p_\alpha-\frac43\\
\beta&2q-1&h&5p_\alpha-\frac43\\
\gamma&4q-3&&
\end{array}
\]
Initially, only report class \(n\) for true type \(\alpha\) is eliminable, so \(p_\alpha\ge1/3\).
Then report \(\ell\) is strictly dominated for true type \(h\), giving \(q\ge1/2\).
Next \(\alpha\) weakly dominates \(n\) for true type \(\beta\), with strict inequality at a surviving profile with \(q=1\), so \(p_\alpha\ge2/3\).
Then \(h\) weakly dominates \(\ell\) for true type \(\ell\), with strict inequality at \(p_\alpha=1\), and hence \(q=1\).
Finally, \(\alpha\) strictly dominates \(n\) for true type \(\gamma\).
At every intermediate step, the reports needed as strictness witnesses survive because their payoff differences still take both signs.
Consequently every exhaustive sequence reaches
\[
\widehat\theta_1(\alpha)=\widehat\theta_1(\beta)=\widehat\theta_1(\gamma)=\alpha,
\qquad
\widehat\theta_2(\ell)=\widehat\theta_2(h)=h
\]
The mechanism then always selects \(y\), which is inefficient at \((\alpha,\ell)\).

Under \(\chi^Z\), reports \(\alpha\) and \(\beta\) of agent 1 are payoff-equivalent; write \(m\) for this class and \(g\) for report \(\gamma\).
Let \(q=\Pr(\widehat\theta_2=h)\) and \(p_\gamma=\Pr(\widehat\theta_1=\gamma)\).
The relevant payoff differences are
\[
\begin{array}{c|c@{\qquad}c|c}
\text{type 1}&U_1(g)-U_1(m)
&\text{type 2}&U_2(h)-U_2(\ell)\\ \hline
\alpha&-1&\ell&-2p_\gamma\\
\beta&1-2q&h&3-5p_\gamma\\
\gamma&3-4q&&
\end{array}
\]
Report \(g\) is strictly dominated by \(m\) for true type \(\alpha\), while report \(h\) is weakly dominated by \(\ell\) for true type \(\ell\).
Before the latter deletion, no report of true type \(\beta\), \(\gamma\), or \(h\) can be the first to disappear: with both reports of type \(\ell\) surviving, \(q=0\) and \(q=1\) make the two agent-1 differences take both signs; with both report classes of types \(\beta\) and \(\gamma\) surviving, \(p_\gamma=0\) and \(p_\gamma=2/3\) make \(3-5p_\gamma\) take both signs.
Thus a profile with \(p_\gamma>0\) remains available, making the dominance for true type \(\ell\) strict, so every exhaustive sequence deletes report \(h\) for that type and obtains \(q\le1/2\).
It follows that \(g\) weakly dominates \(m\) for true type \(\beta\) and strictly dominates it for true type \(\gamma\).
Neither report of true type \(h\) can disappear before both deletions, because surviving profiles with \(p_\gamma\le1/3\) and \(p_\gamma=2/3\) make \(3-5p_\gamma\) take both signs.
After these deletions \(p_\gamma\ge2/3\), so \(h\) is strictly dominated by \(\ell\) for true type \(h\).
Every exhaustive sequence therefore reaches
\[
\widehat\theta_1(\alpha)\in\{\alpha,\beta\},\qquad
\widehat\theta_1(\beta)=\widehat\theta_1(\gamma)=\gamma,
\qquad
\widehat\theta_2(\ell)=\widehat\theta_2(h)=\ell
\]
At the true profile \((\alpha,h)\), every surviving strategy profile induces decision \(z\), although \(y\) is uniquely efficient there.
\end{proof}

The example above establishes the failure for every deterministic efficient policy, but not for randomized tie-breaking at its efficient tie.

\begin{conjecture}
There exists a social-choice environment in which, for every efficient decision policy, including every randomized selection at efficient ties, exhaustive IEWDS under AGV eliminates all efficient strategy profiles, independently of the order of elimination.
\end{conjecture}

\section{VCG mechanisms and GU-VCG}
\label{app:guvcg}

The arbitrary Groves terms permit many payment rules within the VCG family.
The choice \(h_i\equiv0\) is algebraically natural and prior-free.
Another standard prior-free choice is the Clarke pivot rule
\[
 h_i^{\mathrm{piv}}(\widehat\theta_{-i})
 :=-\max_{x\in X}\sum_{j\ne i}u_j(\widehat\theta_j,x)
\]
With our sign convention, its transfer to agent \(i\) is
\[
 y_i^{\VCG,\mathrm{piv}}(\widehat\theta)
 =\sum_{j\ne i}u_j\bigl(\widehat\theta_j,\chi(\widehat\theta)\bigr)
 -\max_{x\in X}\sum_{j\ne i}u_j(\widehat\theta_j,x)
 \le0
\]
If a neutral report assigning zero payoff to every decision were available to agent \(i\), then under that report the chosen decision would maximize the other agents' reported welfare and the pivot transfer would be zero.

Although the VCG family itself is prior-free, a reference distribution may be used ex ante to select a member of the family.
For each fixed reference distribution \(\nu\), one may choose Groves terms \(h_i^\nu(\widehat\theta_{-i})\); the resulting mechanism remains a dominant-strategy implementation whether or not \(\nu=\mu\), in particular, whether or not the true types are independent.

We consider two such choices based on a product reference distribution
\[
  \nu:=\bigotimes_{i\in\N}\nu_i
\]
Write \(V_j^\nu\) and \(\gamma_B^{\nu,i\to j}\) for the objects in Section~1 with \(\mu\) replaced by \(\nu\).

The first prior-dependent choice is the \emph{conditionally centered VCG rule}:
\[
\begin{aligned}
 h_i^{\mathrm{cen},\nu}(\widehat\theta_{-i})
 &:=-\E_{\widetilde\theta_i\sim\nu_i}\Bigl[
 \sum_{j\ne i}u_j\bigl(\widehat\theta_j,
 \chi(\widetilde\theta_i,\widehat\theta_{-i})\bigr)\Bigr]\\
 y_i^{\VCG,\mathrm{cen},\nu}(\widehat\theta)
 &:=y_i^{\VCG,h^{\mathrm{cen},\nu}}(\widehat\theta)
 =\sum_{j\ne i}\gamma_{\N\setminus\{i\}}^{\nu,i\to j}
\end{aligned}
\]
For every fixed \(\widehat\theta_{-i}\), its expected transfer over \(\widetilde\theta_i\sim\nu_i\) is zero.
The second is \emph{guaranteed-utility VCG} (GU-VCG), introduced here.
Under the additional conditions stated below, it is also a GUE implementation in a restricted strategy space.

\subsection{The GU-VCG mechanism}

Define
\begin{equation}
 y_i^{\GUVCG,\nu}(\widehat\theta)
 :=y_i^{\VCG,\mathrm{cen},\nu}(\widehat\theta)
   +V_i^\nu(\varnothing)
   -V_i^\nu(\widehat\theta_{-i})
 \label{eq:GUVCG}
\end{equation}
The adjustment depends only on the other agents' reports.
Hence GU-VCG remains a Groves mechanism, and truthful reporting is a dominant strategy for every reference distribution \(\nu\), whether or not it is correctly calibrated.
GU-VCG is not budget-balanced.\footnote{For two agents, up to agent-specific additive constants, this transfer rule coincides with the externality-free construction of Zik \cite{Zik2021}.}

The individual guarantees below do not make GU-VCG a GUE implementation in the unrestricted reporting game.
Because its transfers need not sum to zero, the agents' aggregate utility can exceed the sum of their guaranteed utilities when the mechanism runs a deficit.
The three-agent example below exhibits such a deviation explicitly.

Suppose now that the product reference distribution is correct, \(\nu=\mu\).
For every fixed \(\widehat\theta_{-i}\), the conditional-centering property of \(y_i^{\VCG,\mathrm{cen},\mu}\) and equation~\eqref{eq:GUVCG} give
\begin{equation}
 \E_{\theta_i}\!\Bigl[
 u_i\bigl(\theta_i,\chi(\theta_i,\widehat\theta_{-i})\bigr)
 +y_i^{\GUVCG,\mu}(\theta_i,\widehat\theta_{-i})
 \Bigr]
 =V_i(\varnothing)
 \label{eq:GUVCGguarantee}
\end{equation}

\begin{proposition}[Restricted GUE property]
Suppose \(\nu=\mu\), and restrict attention to reporting-strategy profiles for which, for every agent \(i\), the other agents' reports \(\widehat\theta_{-i}\) are independent of \(\theta_i\).
Then truthful reporting is a GUE, so GU-VCG is a GUE implementation in this restricted strategy space.
\end{proposition}

\begin{proof}
The restriction allows arbitrary ex-ante coordination and common randomization independent of the realized types, but excludes conditioning reports on information about another agent's realized type.
By equation~\eqref{eq:GUVCGguarantee} and the independence of \(\widehat\theta_{-i}\) from \(\theta_i\), truthful reporting gives agent \(i\) expected utility \(V_i(\varnothing)\) against every admissible strategy profile of the other agents.
Because truthful reporting is a dominant strategy, no admissible reporting strategy can give agent \(i\) expected utility above \(V_i(\varnothing)\).
Thus every admissible strategy profile gives each agent at most her guaranteed utility, while the truthful profile gives every agent exactly that utility.
The truthful profile therefore maximizes total expected utility and is a GUE; it also induces the efficient decision policy.
\end{proof}

Correct calibration and the informational restriction are needed for the utility guarantees and the GUE conclusion; efficiency and dominant-strategy truthfulness require neither.
Thus GU-VCG and TU-GUM provide the same individual utility guarantees in the standard private-type game, but their robustness differs: TU-GUM combines its guarantees with budget balance, whereas GU-VCG is a GUE only under the restriction above.

\subsection{The three-agent example}

Consider the example in Section~\ref{sec:three-agent-example}.
Here \(V_i(\varnothing)=3\), while the expected payoff of agent \(i\), with the other two reports fixed and \(i\)'s type drawn from the prior, is
\[
 V_i(\widehat\theta_{-i})
 =\begin{cases}
 0,&\widehat\theta_{-i}=(-6,-6)\\
 5,&\widehat\theta_{-i}\in\bigl\{(-6,10),(10,-6)\bigr\}\\
 2,&\widehat\theta_{-i}=(10,10)
 \end{cases}
\]
The GU-VCG adjustment to the conditionally centered rule is therefore \(3,-2,1\), respectively.
More generally, every symmetric VCG rule is determined by three constants \(a,b,c\): for a fixed agent, the transfers corresponding to own reports \(-6,10\) are \(a,a\) when the other reports are \((-6,-6)\), \(b,b+4\) when they are mixed, and \(c,c\) when they are \((10,10)\).
Taking \(\nu=\mu\) for the two reference-distribution rules, the table displays the full symmetric VCG family, the four VCG choices described above, and the two budget-balanced mechanisms; as before, permuting a report profile permutes the coordinates of its transfer vector.

\begin{table}[htbp]
\centering
\caption{The symmetric VCG family and six specific transfer rules.}
\label{tab:mechanism-summary-extended}
\small
\renewcommand{\arraystretch}{1.12}
\setlength{\tabcolsep}{3.2pt}
\begin{tabularx}{\textwidth}{@{}l *{4}{>{\centering\arraybackslash}X}@{}}
\toprule
Reports \(\widehat\theta_{\N}\)
 & \(\textcolor{cblue}{\tuplefit{-6}{-6}{-6}{-6}{-6}{-6}}\)
 & \(\textcolor{cblue}{\tuplefit{-6}{-6}{10}{-6}{-6}{10}}\)
 & \(\textcolor{cblue}{\tuplefitcompact{-6}{b+4}{b+4}{-6}{10}{10}}\)
 & \(\textcolor{cblue}{\tuplefit{20}{20}{20}{10}{10}{10}}\)\\
\midrule
Decision
 & \(\textcolor{dpurple}{x=0}\)
 & \(\textcolor{dpurple}{x=0}\)
 & \(\textcolor{dpurple}{x=1}\)
 & \(\textcolor{dpurple}{x=1}\)\\
\specialrule{\lightrulewidth}{0.35em}{0pt}
sym.\ VCG & \(\tuplefit{-6}{-6}{-6}{a}{a}{a}\) & \(\tuplefit{-6}{-6}{10}{b}{b}{a}\) & \(\tuplefitcompact{-6}{b+4}{b+4}{c}{b+4}{b+4}\) & \(\tuplefit{20}{20}{20}{c}{c}{c}\)\\
VCG, \(h=0\) & \(\tuplefit{-6}{-6}{-6}{0}{0}{0}\) & \(\tuplefit{-6}{-6}{10}{0}{0}{0}\) & \(\tuplefitcompact{-6}{b+4}{b+4}{20}{4}{4}\) & \(\tuplefit{20}{20}{20}{20}{20}{20}\)\\
VCG\(^{\mathrm{piv}}\) & \(\tuplefit{-6}{-6}{-6}{0}{0}{0}\) & \(\tuplefit{-6}{-6}{10}{-4}{-4}{0}\) & \(\tuplefitcompact{-6}{b+4}{b+4}{0}{0}{0}\) & \(\tuplefit{20}{20}{20}{0}{0}{0}\)\\
VCG\(^{\mathrm{cen}}\) & \(\tuplefit{-6}{-6}{-6}{0}{0}{0}\) & \(\tuplefit{-6}{-6}{10}{-2}{-2}{0}\) & \(\tuplefitcompact{-6}{b+4}{b+4}{0}{2}{2}\) & \(\tuplefit{20}{20}{20}{0}{0}{0}\)\\
\rowcolor{llgreen}
GU-VCG & \(\tuplefit{-6}{-6}{-6}{3}{3}{3}\) & \(\tuplefit{-6}{-6}{10}{-4}{-4}{3}\) & \(\tuplefitcompact{-6}{b+4}{b+4}{1}{0}{0}\) & \(\tuplefit{20}{20}{20}{1}{1}{1}\)\\
\faintdottedrule
AGV & \(\tuplefit{-6}{-6}{-6}{0}{0}{0}\) & \(\tuplefit{-6}{-6}{10}{-1}{-1}{2}\) & \(\tuplefitcompact{-6}{b+4}{b+4}{-2}{1}{1}\) & \(\tuplefit{20}{20}{20}{0}{0}{0}\)\\
\faintdottedrule
sym.\ TU-GUM
 & \(\tuplefit{-6}{-6}{-6}{0}{0}{0}\) & \(\tuplefit{-6}{-6}{10}{-3}{-3}{6}\) & \(\tuplefitcompact{-6}{b+4}{b+4}{2}{-1}{-1}\) & \(\tuplefit{20}{20}{20}{0}{0}{0}\)\\
\bottomrule
\end{tabularx}
\end{table}

For GU-VCG, for each fixed profile of the other reports, a truthful agent's utility averaged over his own type is
\begin{align*}
 (10,10):&\quad \frac{(-6+1)+(10+1)}2=3\\
 (-6,10):&\quad \frac{(0-4)+(10+0)}2=3\\
 (-6,-6):&\quad \frac{(0+3)+(0+3)}2=3
\end{align*}
These calculations instantiate the restricted GUE result above.

For the mechanisms displayed in Table~\ref{tab:mechanism-summary-extended}, the corresponding own-type averages are:

\begin{table}[htbp]
\centering
\caption{Utility of a truthful agent, averaged over his own type, for fixed reports of the others.}
\label{tab:conditional-utilities-extended}
\small
\begin{tabularx}{\textwidth}{@{}>{\raggedright\arraybackslash}X *{3}{>{\centering\arraybackslash}X}@{}}
\toprule
Mechanism & \((10,10)\) & \((-6,10)\) or \((10,-6)\) & \((-6,-6)\)\\
\midrule
sym.\ VCG & \(c+2\) & \(b+7\) & \(a\)\\
VCG, \(h=0\) & 22 & 7 & 0\\
VCG\(^{\mathrm{piv}}\) & 2 & 3 & 0\\
VCG\(^{\mathrm{cen}}\) & 2 & 5 & 0\\
\rowcolor{llgreen}
GU-VCG & 3 & 3 & 3\\
\faintdottedrule
AGV & 1 & 5 & 1\\
\faintdottedrule
sym.\ TU-GUM & 3 & 3 & 3\\
\bottomrule
\end{tabularx}
\end{table}
\FloatBarrier

The identical GU-VCG and symmetrized TU-GUM rows reflect their common individual utility guarantees under the informational restriction above.
Unlike TU-GUM, GU-VCG is not budget-balanced; in the four columns of Table~\ref{tab:mechanism-summary-extended}, its transfer sums are \(9,-5,1,3\).

The informational restriction is essential.
If the grand coalition shares its realized types, it can have everyone report \(-6\) when at most one type is \(10\), and everyone report \(10\) otherwise.
Thus the same decision is made in every state as under truthful reporting, but instead of transfers with zero expected value, every agent receives a transfer of either \(1\) or \(3\).
The resulting expected total utility is
\[
 \frac18\,9+\frac38\,9+\frac38\,17+\frac18\,33=15
\]
whereas truthful reporting gives expected total utility \(9\), the sum of the three guarantees.
This example also demonstrates that the GUE property of GU-VCG, and hence the collusion-proofness implied by it, holds only when no agent can condition his report on another agent's realized type.

\section[The three-agent social choice example in the Project Management Mechanism]{The three-agent social choice example in the\\Project Management Mechanism}
\label{app:pmmstatic}

This appendix embeds the three-agent social choice example in the Project Management Mechanism and shows how a particular strategy of the principal reproduces the fixed-order TU-GUM.
The construction concerns the order \(1,2,3\), not the symmetrized payment rule.

The agents' initial types are fixed as part of the Project Management Mechanism specification.
The random variable denoted by \(\theta_i\in\{-6,10\}\) in the static formulation is represented here as an equiprobable private chance event realized before the reporting stage.
The contractible result spaces are
\[
 \mathcal R_0=\{0,1\}
 \qquad r_0=x
 \qquad \mathcal R_i=\{\ast\}\quad(i=1,2,3)
\]
so the social choice component of the result vector is precisely the decision \(x\).
The principal's payoff is zero, \(u_0(r)=0\), and, writing \(\theta_i\) for the realization recorded in \(\xi_i\), agent \(i\)'s payoff is
\[
 u_i(\xi_i,r)=r_0\theta_i=x\theta_i
\]

For this preference and protected expected utility level \(3\), the agent's truthful contract offer can be presented in a much simpler form.
Directly applying the definition of the truthful offer permits a richer set of protocols, including additional contractible messages before the private chance event is realized.
After discarding the options that are never worth choosing for the principal, the offer can be reduced to the following one-parameter family.
This is a reduction of the principal's relevant choices, not a pointwise characterization of all protocols admitted by the offer.

The contract protocol gives the principal a specified reporting window.
Immediately before requesting agent \(i\)'s report, she chooses and announces a scalar \(a\in\mathbb R\).
If the agent reports \(s\in\{-6,10\}\) and the mechanism selects \(r_0=x\), his transfer is
\begin{equation}
 t_i^a(-6,x)=a+6x
 \qquad
 t_i^a(10,x)=6-a-10x
 \label{eq:pmm-scalar-rule}
\end{equation}
Thus, in the order \((-6,0),(-6,1),(10,0),(10,1)\), the four transfers are
\[
 a,\qquad a+6,\qquad 6-a,\qquad -4-a
\]
The scalar announcement and the agent's report are contractible messages and are therefore part of \(m\).
The offer does not prescribe when within the reporting window the principal requests the report, which scalar she announces, or how the mechanism subsequently selects the social choice.

If the private chance event is \(\theta_i=s\) and the agent reports truthfully, then \eqref{eq:pmm-scalar-rule} gives
\[
 -6x+t_i^a(-6,x)=a
 \qquad
 10x+t_i^a(10,x)=6-a
\]
independently of the social choice.
Conditional on any preceding history, the two realizations remain equiprobable, so truthful reporting guarantees expected utility
\[
 1/2 \cdot a+1/2 \cdot (6-a)=3
\]
Thus every member of the scalar family is admitted by the truthful offer.

\Needspace{15\baselineskip}
To reproduce fixed-order TU-GUM, Table~\ref{tab:pmm-tu-gum-scalars} assigns a scalar \(a_i^h\) to each earlier-report history \(h=(\widehat\theta_1,\ldots,\widehat\theta_{i-1})\).
\begin{table}[h!]
\centering
\caption{Scalar announcements reproducing fixed-order TU-GUM.}
\label{tab:pmm-tu-gum-scalars}
\renewcommand{\arraystretch}{1.13}
\begin{tabular}{@{}c c c@{}}
\toprule
Agent & Earlier reports \(h\) & Announced scalar \(a_i^h\)\\
\midrule
1 & \historycenter{\(\varnothing\)} & \(\alignednum{-2}{.5}\)\\
\midrule
2 & \historycenter{\(-6\)} & \(\alignednum{-0}{.5}\)\\
2 & \historycenter{\(10\)} & \(\alignednum{-4}{.5}\)\\
\midrule
3 & \historyleft{\(\pairfit{-6}{-6}{-6}{-6}\)} & \(\alignednum{3}{}\)\\
3 & \historyleft{\(\pairfit{-6}{10}{-6}{10}\) or \(\pairfit{10}{-6}{10}{-6}\)} & \(\alignednum{-4}{}\)\\
3 & \historyleft{\(\pairfit{-6}{10}{10}{10}\)} & \(\alignednum{-5}{}\)\\
\bottomrule
\end{tabular}
\end{table}
\FloatBarrier
At each such history, the principal announces the listed scalar immediately before requesting agent \(i\)'s report.
The announcement is fixed and cannot subsequently be revised.
After all reports, the mechanism selects, under this strategy,
\[
 r_0=x=\chi(\widehat\theta)
\]
and pays \(t_i^{a_i^h}(\widehat\theta_i,x)\) to agent \(i\).

For example, after \(\widehat\theta_1=-6\), the principal announces \(a_2^{(-6)}=-0.5\).
If agent 2 reports \(10\), his transfer is \(6.5\) when \(r_0=0\), and \(-3.5\) when \(r_0=1\).
These are exactly his fixed-order TU-GUM transfers at \((-6,10,-6)\) and \((-6,10,10)\), respectively.
Checking the other rows in the same way reproduces the fixed-order transfer vector \(y^{\GUM}\) at every report profile.

The maximin argument applies to the full truthful offers, not only to their scalar reductions.
By the definition of a truthful offer, each agent has a truthful strategy guaranteeing expected utility at least \(3\), whatever the principal and the other agents do.
If all three agents use such strategies, their expected utilities sum to at least \(9\) against any strategy of the principal.
Because transfers cancel in total utility and the expected total payoff cannot exceed \(9\), the principal's expected utility against this response is at most \(0\).
Hence no strategy of the principal can guarantee her more than \(0\).
The scalar strategy above reproduces the budget-balanced fixed-order TU-GUM, so the principal receives utility \(0\) at every report profile.
It is therefore a maximin strategy even when the full truthful offers are available, and truthful reporting generates the efficient TU-GUM outcome.

The symmetrized rule is not represented by this sequential scalar construction: an early agent's transfer may depend on later reports even when his own report and the social choice are held fixed.


\begin{thebibliography}{99}
\small
\setlength{\itemsep}{-0.06em}
\setlength{\parskip}{0pt}

\bibitem{Arrow1979}
Kenneth J. Arrow.
\newblock The property rights doctrine and demand revelation under incomplete
information.
\newblock In Michael J. Boskin, editor, \emph{Economics and Human Welfare:
Essays in Honor of Tibor Scitovsky}, pages 23--39.
Academic Press, New York, 1979.

\bibitem{AtheySegal2013}
Susan Athey and Ilya Segal.
\newblock An efficient dynamic mechanism.
\newblock \emph{Econometrica}, 81(6):2463--2485, 2013.

\bibitem{BergemannValimaki2010}
Dirk Bergemann and Juuso V\"alim\"aki.
\newblock The dynamic pivot mechanism.
\newblock \emph{Econometrica}, 78(2):771--789, 2010.

\bibitem{Clarke1971}
Edward H. Clarke.
\newblock Multipart pricing of public goods.
\newblock \emph{Public Choice}, 11:17--33, 1971.

\bibitem{Csoka2015}
Endre Cs\'oka.
\newblock Efficient teamwork.
\newblock \href{https://arxiv.org/abs/cs/0602009}{arXiv:cs/0602009}.
First version: 2006; latest version: 2018.

\bibitem{Csoka2021}
Endre Cs\'oka.
\newblock A robust efficient dynamic mechanism.
\newblock arXiv:2110.15219, 2021.

\bibitem{CsokaPriorFree2025}
Endre Cs\'oka.
\newblock Prior-free collusion-proof dynamic mechanisms.
\newblock arXiv:2511.15727, 2025; revised July 2026.

\bibitem{CsokaPongraczRodivilov2025}
Endre Cs\'oka, Andr\'as Pongr\'acz, and Alexander Rodivilov.
\newblock Guaranteed utility equilibrium.
\newblock Forthcoming in the \emph{Proceedings of the 27th ACM Conference on Economics and Computation (EC '26)}, 2026.

\bibitem{CsokaEtAl2026}
Endre Cs\'oka, Heng Liu, Alexander Rodivilov, and Alexander Teytelboym.
\newblock Efficiency and collusion-proofness in dynamic mechanism design.
\newblock Working paper, July 31, 2026.

\bibitem{DAspremontGerardVaret1979}
Claude d'Aspremont and Louis-Andr\'e G\'erard-Varet.
\newblock Incentives and incomplete information.
\newblock \emph{Journal of Public Economics}, 11(1):25--45, 1979.

\bibitem{GreenLaffont1977}
Jerry R. Green and Jean-Jacques Laffont.
\newblock Characterization of satisfactory mechanisms for the revelation of preferences for public goods.
\newblock \emph{Econometrica}, 45(2):427--438, 1977.

\bibitem{GreenLaffont1979}
Jerry R. Green and Jean-Jacques Laffont.
\newblock \emph{Incentives in Public Decision-Making}.
\newblock North-Holland, Amsterdam, 1979.

\bibitem{Groves1973}
Theodore Groves.
\newblock Incentives in teams.
\newblock \emph{Econometrica}, 41(4):617--631, 1973.

\bibitem{Myerson1979}
Roger B. Myerson.
\newblock Incentive compatibility and the bargaining problem.
\newblock \emph{Econometrica}, 47(1):61--73, 1979.

\bibitem{Vickrey1961}
William Vickrey.
\newblock Counterspeculation, auctions, and competitive sealed tenders.
\newblock \emph{The Journal of Finance}, 16(1):8--37, 1961.

\bibitem{Zik2021}
Boaz Zik.
\newblock Ex-post implementation with social preferences.
\newblock \emph{Social Choice and Welfare}, 56(3):467--485, 2021.

\end{thebibliography}
\end{document}